\documentclass[journal,10pt]{IEEEtran}
\usepackage{amsmath}
\usepackage{bbm}
\usepackage{graphicx}
\usepackage{epsfig}
\usepackage{array}
\usepackage{psfrag}
\usepackage{amssymb}
\usepackage{mathdots}
\usepackage{color}
\usepackage{pstricks,pst-node,pst-text,pst-3d,pst-plot}
\usepackage{psfrag}
\usepackage{enumerate}
\usepackage{url,cite}
\usepackage{amsfonts}%
\usepackage{amsthm}
\usepackage{dsfont}%
\usepackage{verbatim}%
\usepackage{setspace}
\usepackage{float}
\usepackage{url}
\usepackage{fancyhdr}
\usepackage{bm}
\usepackage{lipsum} 
\newcommand\blfootnote[1]{%
  \begingroup
  \renewcommand\thefootnote{}\footnote{#1}%
  \addtocounter{footnote}{-1}%
\endgroup
}
\usepackage{algorithm}
\usepackage{algorithmic}

\newtheorem{Def}{Definition}

\newcommand{\Dikm}{D_i(\bfm,k)}
\newcommand{\Dik}{D_i(k)}
\newcommand{\bfm}{{\bm m}}
\newcommand{\parmdef}[1]{\triangleq [#1_1(\bfm),\cdots,#1_N(\bfm)]^T}

\newcommand{\SNRkst}{\script{S}_{\rm NR}^*(k)}
\usepackage{cleveref}
\newcommand{\Bmax}{{B_{\rm max}}}
\newcommand{\figures}{figures}

\newcommand{\Paux}{P_{\rm X}}
\newcommand{\Psiaux}{\Psi_{\rm X}}
\newcommand{\sK}{\script{K}_K(l)}
\newcommand{\Bistl}{\tilde{B}_i^{(l)}}
\newcommand{\Bikst}{\tilde{B}_i(k)}
\newcommand{\Bik}{B_i(k)}
\newcommand{\aik}{a_i(k)}
\newcommand{\Lag}{L_{\rm agr}}
\newcommand{\PRopt}{\tilde{P}_i(k)}
\newcommand{\muRopt}{\tilde{\mu}_i(k)}
\newcommand{\Ropt}{\tilde{R}_i(k)}
\newcommand{\bRopt}{R^{\rm (opt)}_i}
\newcommand{\Oneopt}{\tilde{\mathds{1}}_i(k)}
\newcommand{\Psiopt}{\tilde{\Psi}(k)}

\newcommand{\Oneik}{\mathds{1}_i(k)}
\newcommand{\Ai}{A_i}

\newcommand{\pim}{\pi_{\bfm}}
\newcommand{\phimax}{\phi_{\rm max}}
\newcommand{\Rik}{R_i(k)}

\newcommand{\Qik}{Q_i(k)}

\newcommand{\Yik}{Y_i(k)}
\newcommand{\bfYk}{{\mathbf Y}(k)}
\newcommand{\bfQk}{{\mathbf Q}(k)}
\newcommand{\LW}[1]{W_0\lb #1e^{-1}\rb}
\newcommand{\gammaikm}{\gamma_i(\bfm)}
\newcommand{\gammaik}{\gamma_i(k)}
\newcommand{\lambdaNR}{\bm{\lambda}_{{\rm NR}}}

\newcommand{\NR}{N_{\rm R}}
\newcommand{\NNR}{N_{\rm NR}}

\newcommand{\sM}{\script{M}}
\newcommand{\sN}{\script{N}}
\newcommand{\Nk}{N_k}
\newcommand{\sNk}{\script{N}_k}
\newcommand{\sNR}{\script{N}_{\rm R}}

\newcommand{\sNNR}{\script{N}_{\rm NR}}

\newcommand{\sR}{\script{R}_{\rm Lamb}}
\newcommand{\sS}{\script{S}}

\newcommand{\SRT}{\script{S}_{\rm RT}}

\newcommand{\SRk}{\script{S}_{\rm R}\lb k\rb}
\newcommand{\SRstk}{\script{S}_{\rm R}^*\lb k\rb}

\newcommand{\SRkst}{\script{S}_{\rm R}^*\lb k\rb}
\newcommand{\SNRk}{\script{S}_{\rm NR}\lb k\rb}
\newcommand{\bfPRk}{\bfP\lb k\rb}

\newcommand{\PRikm}{P_i\lb \bfm,k\rb}
\newcommand{\PRik}{P_i\lb k\rb}
\newcommand{\PRikst}{P_i^*\lb k\rb}

\newcommand{\PNRik}{P_i\lb k\rb}
\newcommand{\bmuRi}{\overline{R}_i}

\newcommand{\bfmuRk}{{\bm \mu}\lb k\rb}

\newcommand{\muRikm}{\mu_i\lb \bfm,k\rb}
\newcommand{\muRik}{\mu_i\lb k\rb}

\newcommand{\muRistk}{\mu_{\ist}\lb k\rb}

\newcommand{\parPRdef}[1]{\triangleq \left[#1_i(k)\right]_{i\in\sN}}
\newcommand{\parPNRdef}[1]{\triangleq \left[#1_i(k)\right]_{i\in\sNNR}}

\newcommand{\PsiRik}{\Psi_{\rm R}(i,k)}
\newcommand{\PsiRikst}{\Psi_{\rm R}^*(i,k)}
\newcommand{\PsiNRik}{\Psi_{\rm NR}(i,k)}
\newcommand{\PsiNRikst}{\Psi_{\rm NR}^*(i,k)}
\newcommand{\PsiNRistkst}{\Psi_{\rm NR}^*(\ist,k)}

\newcommand{\ist}{i_{\rm NR}^*}

\newcommand{\bfrk}{{\bm r}\lb k\rb}

\newtheorem{thm}{Theorem}
\newtheorem{lma}{Lemma}
\DeclareMathOperator{\E}{\mathbb{E}}
\newenvironment{proofsketch}{\par{\it Proof Sketch:}}{\qed\par}

\newcommand{\lb}{\left (}
\newcommand{\rb}{\right )}

\newcommand{\script}[1]{{\mathcal {#1}}}
\newcommand{\Pavg}{P_{\rm avg}}
\newcommand{\Pmax}{P_{\rm max}}

\newcommand{\EE}[1]{\E \left[ #1 \right]}
\newcommand{\EEU}[1]{\E_{\bfU(k)} \left[ #1 \right]}

\newcommand{\bfP}{{\bf P}}

\newcommand{\bgamma}{\overline{\gamma}}

\newcommand{\bfQ}{{\bf Q}}

\newcommand{\bfY}{{\bf Y}}
\newcommand{\bfU}{{\bf U}}

\newcommand{\parRdef}[1]{\triangleq [#1_1(k),\cdots,#1_{\NR}(k)]^T}
\newcommand{\parNRdef}[1]{\triangleq [#1_1(k),\cdots,#1_{\NNR}(k)]^T}

\newcommand{\Rmax}{R_{\rm max}}
\newcommand{\gammamax}{\gamma_{\rm max}}

\newcommand{\Ts}{T}

\begin{document}
\title{Optimal Power Control and Scheduling for Real-Time and Non-Real-Time Data}

\author{Ahmed Ewaisha, Cihan Tepedelenlio\u{g}lu\\
\small{School of Electrical, Computer, and Energy Engineering, Arizona State University, USA.}\\
\small{Email:\{ewaisha, cihan\}@asu.edu}\\
}
\maketitle
\blfootnote{The work in this paper has been supported by NSF Grant ECCS-1307982.}
\blfootnote{Parts of this work have been accepted in IEEE WCNC 2017 conference \cite{Ewai1703:Power}.}
\begin{abstract}
We consider a  joint scheduling-and-power-allocation problem of a downlink cellular system. The system consists of two groups of users: real-time (RT) and non-real-time (NRT) users. Given an average power constraint on the base station, the problem is to find an algorithm that satisfies the RT hard deadline constraint and NRT queue stability constraint. We propose two sum-rate-maximizing algorithms that satisfy these constraints as well as achieving the system's capacity region. In both algorithms, the power allocation policy has a closed-form expression for the two groups of users. However, interestingly, the power policy of the RT users differ in structure from that of the NRT users. The first algorithm is optimal for the on-off channel model with a polynomial-time scheduling complexity in the number of RT users. The second, on the other hand, works for any channel fading model which is shown, through simulations, to have an average complexity that is close-to-linear. We also show the superiority of the proposed algorithms over existing approaches using extensive simulations.
\end{abstract}

\section{Introduction}
Quality-of-service-based scheduling has received much attention recently. It is shown in \cite{Lai20131689} and \cite{piro2011two} that quality-of-service-aware scheduling results in a better performance in LTE systems compared to best-effort techniques. Depending on the application, quality-of-service (QoS) metrics capture long-term throughput \cite{6848162}, short-term throughput \cite{hsiehheavy}, per-user average delay \cite{ewaisha2015joint}, average number of packets missing a specific deadline \cite{A_Theory_of_QoS}, or the average time a user waits to receive its data \cite{hou2015qoe}. Real-time audio and video applications require algorithms that take hard deadlines into consideration. This is because if a real-time packet is not transmitted on time, the corresponding user might experience intermittent connectivity of its audio or video.


The problem of scheduling for wireless systems under hard-deadline constraints has been widely studied in the literature (see, e.g., \cite{hou2011survey} and \cite{radhakrishnan2016review} for a survey).
 In \cite{A_Theory_of_QoS} the authors consider binary erasure channels and present a sufficient and necessary condition to determine if a given problem is feasible. The work is extended in three different directions. The first direction studies the problem under delayed feedback \cite{piro2011two}. The second considers general channel fading models \cite{hou2010scheduling}. 
The third studies multicast video packets that have strict deadlines and utilize network coding to improve the overall network performance \cite{Hou:2015:BDT:2823437.2823441,Adaptive_NC_Deadline}. Unlike the time-framed assumption in the previous works, the authors of \cite{kang2013performance} assume that arrivals and deadlines do not have to occur at the edges of a time frame. They present a scheduling algorithm under the on-off channel fading model and present its achievable region under general arrivals and deadline patterns but with a fixed power transmission. 
In \cite{Elastic_Inelastic} the authors study the scheduling problem in the presence of real-time and non-real-time data. Unlike real-time data, non-real-time data do not have strict deadlines but have an implicit stability constraint on the queues. Using the dual function approach, the problem was decomposed into an online algorithm that guarantees network stability and satisfies the real-time users' constraint.


Power allocation has not been considered for RT users in the literature, to the best of our knowledge. In this paper, we study a throughput maximization problem in a downlink cellular system serving RT and NRT users simultaneously. We formulate the problem as a joint scheduling-and-power-allocation problem to maximize the sum throughput of the NRT users subject to an average power constraint on the base station (BS), as well as a delivery ratio requirement constraint for each RT user. The delivery ratio constraint requires a minimum ratio of packets to be transmitted by a hard deadline, for each RT user. Perhaps the closest to our work are references \cite{Elastic_Inelastic} and \cite{Ewaisha_TVT2015}. The former does not consider power allocation, while the latter assumes that only one user can be scheduled per time slot. The contributions in this paper are as follows:
\begin{itemize}
	\item We present two scheduling-and-power-allocation algorithms. The first is for the on-off channel fading model while the second is for the continuous channel fading model.
	\item We show that both algorithms are optimal. That is, both satisfy the average power constraint, the delivery ratio requirement constraint, in addition to achieving the capacity region. However, the complexity of the first is polynomial in the number of users, while the second is shown to have an average complexity that is close-to-linear.
	\item We present closed-form expressions for the power allocation policy used by both algorithms. It is shown that the power allocation expressions for the RT and NRT users have a different structure.
	\item Through simulations, we show the complexity and throughput performances of the proposed algorithms over baseline ones.
\end{itemize}

The rest of this paper is organized as follows. In Section \ref{Model} we present the system model and the underlying assumptions. The problem is formulated in Section \ref{Problem_Formulation}. For the on-off channel model, the proposed power-allocation and scheduling algorithm as well as its optimality is presented in Section \ref{Proposed_Algorithm}. In Section \ref{Continuous_Fading} we present the optimal algorithm for the continuous channel model as well as another optimal algorithm with a lower complexity. The capacity region of the problem is presented in Section \ref{Capacity_Region}. Simulation results and comparisons with baseline approaches is presented in Section \ref{Results}. Finally, the paper is concluded in Section \ref{Conclusion}.


\section{System Model}
\label{Model}
We assume a time slotted downlink system with slot duration $\Ts$ seconds. The system has a single base station (BS) having access to a single frequency channel. The interference coming from all other neighboring BSs is assumed to be treated as noise. There are $N$ users in the system indexed by the set $\sN\triangleq\{1, \cdots,N\}$. The set of users is divided into the RT users $\sNR\triangleq\{1,\cdots,\NR\}$, and NRT users $\sNNR\triangleq\{\NR+1,\cdots, N\}$ with $\NR$ and $\NNR\triangleq N-\NR$ denoting the number of RT and NRT users, respectively.

We model the channel between the BS and the $i$th user as a fading channel with power gain $\gammaik$. The distribution and statistics of $\gammaik$ is arbitrary and need not be known to the BS nor to any of the users. In this paper, we present the problem for the on-off channel fading case in Sections \ref{Problem_Formulation} and \ref{Proposed_Algorithm} and then we generalize this to the continuous fading case in Section \ref{Continuous_Fading}. The on-off model \cite{A_Theory_of_QoS} corresponds to the well-known binary erasure channel model and models whether the channel is in outage or not. While the continuous fading model is more general and captures all independent and identically distributed channel distributions its solution, as will be seen, has a higher complexity.

For the on-off channel model, if channel $i$ is in a non-outage state during the $k$th slot then $\gammaik=1$, otherwise $\gammaik=0$. Channel gains are fixed over the whole slot and change independently in subsequent slots and are independent across users. Hence, the channel gain follows a Bernoulli process. Channels with a more general fading model will be discussed in Section \ref{Continuous_Fading}. Moreover, $\gammaik$ is known to the BS, for all $i\in\script{N}$, at the beginning of the each slot.

\subsection{Packet Arrival Model}
Let $a_i(k)\in\{0,1\}$ be the indicator of a packet arrival for user $i\in\sN$ at the beginning of the $k$th slot. $\{a_i(k)\}$ is assumed to be a Bernoulli process with rate $\lambda_i$ packets per slot and assumed to be independent across all users in the system. Packets arriving at the BS for the RT users are called real-time packets. RT packets have a strict transmission deadline. If an RT packet is not transmitted by this deadline, this packet is dropped out of the system and does not contribute towards the throughput of the user. However, RT user $i$ is satisfied if it receives, on average, more than $q_i\%$ of its total number of packets. We refer to this constraint as the QoS constraint for user $i$. Here we assume that real-time packets arriving at the beginning of the $k$th slot have their deadline at the end of this slot.

On the other hand, packets arriving to the BS for the NRT users can be transmitted at any point in time. Thus, packets for NRT user $i$ are stored, at the BS, at user $i$'s  (infinite-sized \cite{Bertsekas_Data_Networks}) buffer and served on a first-come-first-serve basis. Since the arrival rate $\lambda_i$, for NRT user $i$, might be higher than what the system can support, we define $r_i(k)$ as an admission controller for user $i$ at slot $k$. At the beginning of slot $k$, the BS sets $r_i(k)$ to $1$ if the BS decides to admit user $i$'s arrived packet to the buffer, and to $0$ otherwise. The time-average number of packets admitted to user $i$'s buffer is
\begin{equation}
\Ai\triangleq \limsup_{K\rightarrow \infty}\frac{1}{K}\sum_{k=1}^K \EE{r_i(k)}, \hspace{0.25in} i\in\sNNR.
\label{Avg_Admit}
\end{equation}
And the queue associated with NRT user $i$ is given by
\begin{equation}
Q_i(k+1)=\lb \Qik + Lr_i(k)-\muRik\Rik\rb^+, \hspace{0.05in} i\in\sNNR,
\label{Queues}
\end{equation}
where $r_i(k)$ is the admission control decision variable for NRT user $i$ at the beginning of slot $k$. We note that no admission controller is defined for the RT users since their buffers cannot build up due to the presence of a deadline.

\subsection{Service Model}
\begin{figure}%
\centering
\includegraphics[width=1\columnwidth]{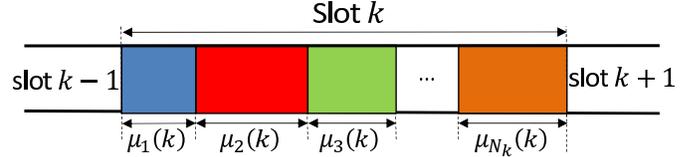}%
\caption{In the $k$th time slot, the BS chooses $\Nk$ users to be scheduled. All time slots have a fixed duration of $\Ts$ seconds.}%
\label{Time_Slot}%
\end{figure}
Following \cite{hou2010scheduling} we assume that more than one user can be scheduled in one time slot. However, due to the existence of a single frequency channel in the system, the BS transmits to the scheduled users sequentially as shown in Fig. \ref{Time_Slot}. At the beginning of the $k$th slot, the BS selects a set of RT users denoted by $\SRk\subseteq\sNR$ and a set of NRT users $\SNRk\subseteq\sNNR$ to be scheduled during slot $k$. Thus a total of $\Nk\triangleq\vert\sNk\vert$ users are scheduled at slot $k$ where $\sNk\triangleq\SRk\cup\SNRk$ (Fig. \ref{Time_Slot}). Moreover, the BS assigns an amount of power $\PRik$ for every user $i\in\sNk$. This dictates the transmission rate for each user according to the channel capacity given by
\begin{equation}
\Rik=\log \lb 1+\PRik\gammaik\rb.
\label{Rate_On_Off}
\end{equation}
Finally, the BS determines the duration of time, out of the $\Ts$ seconds, that will be allocated for each scheduled user. We define the variable $\muRik$ to represent the duration of time, in seconds, assigned for user $i\in\sN$ during the $k$th slot (Fig. \ref{Time_Slot}). Hence, $\muRik\in[0,\Ts]$ for all $i\in\sN$. The BS decides the value of $\muRik$ for each user $i\in\sN$ at the beginning of slot $k$. Since RT users have a strict deadline, then if an RT user is scheduled at slot $k$, then it should be allocated the channel for a duration of time that allows the transmission of the whole packet. Thus we have
\begin{equation}
\muRik=\left\{
\begin{array}{lll}
	\frac{L}{\Rik} &\mbox{ if }i\in\SRk\\
	0 &\mbox{ if }i\in\sNR\backslash\SRk
\end{array}
\right.,
\label{Num_Slots}
\end{equation}
where $L$ is the number of bits per packet, that is assumed to be fixed for all packets in the system. The extension to multiple packet types of different lengths will be addressed in Section \ref{General_L}. Equation \eqref{Num_Slots} means that, depending on the transmission power, if RT user $i$ is scheduled at slot $k$, then it is assigned as much time as required to transmit its $L$ bits. Hence, unlike the NRT users that have $\muRik\in[0,T]$, $\muRik$ is further restricted to the set $\{0,L/\Rik\}$ for the RT users. For ease of presentation, we denote $\bfQ(k)\parNRdef{Q}$. In the next section we present the problem formally.

\section{Problem Formulation for On-Off Channels}
\label{Problem_Formulation}
We are interested in finding the scheduling and power allocation algorithm that maximizes the sum-rate of all NRT users subject to the system constraints. In this paper we restrict our search to slot-based algorithms which, by definition, take the decisions only at the beginning of the time slots.

Now define the average rate of user $i\in\sNNR$ to be $\bmuRi\triangleq \liminf_{K\rightarrow \infty}\sum_{k=1}^K\muRik \Rik/(L\Ts K)$ packets per slot. Thus the problem is to find the scheduling, power allocation and packet admission decisions at the beginning of each slot, that solve the following problem
\begin{align}
\label{Prob_DL}
\underset{\{\bfmuRk,\bfPRk,\bfrk\}_{k=1}^\infty}{\text{maximize}} &\sum_{i\in\sNNR}\bmuRi,&\\
\text{subject to } & r_i(k)\leq a_i(k), \hspace{0.1in}\forall i\in\sNNR,
\label{Admission_Decision}\\
& \limsup_{k\rightarrow\infty}\EE{\Qik}<\infty \hspace{0.1in}\forall i\in\sNNR,
\label{NRT_QoS}\\
&\bmuRi\geq\lambda_i q_i, \hspace{0.1in} \forall i\in\sNR,
\label{RT_QoS}\\
&\limsup_{K\rightarrow \infty}\frac{1}{K\Ts}\sum_{k=1}^K \sum_{i\in\sN}\PRik \muRik\leq \Pavg,
\label{P_avg}\\
& 0\leq\PRik\leq \Pmax, \hspace{0.1in}\forall i\in\sN,
\label{P_max}\\
&\sum_{i\in\sN}\muRik= \Ts \hspace{0.1in}\forall k\geq 1,
\label{Single_Tx_at_a_Time}\\
&0\leq\muRik\leq\Ts \hspace{0.1in} \forall i\in\sN,
\label{muRik_Constr_NRT}
\end{align}
where $\bfmuRk\parPRdef{\mu}$, $\bfPRk\parPRdef{P}$ while $\bfrk\parPNRdef{r}$. Constraint \eqref{Admission_Decision} says that no packets should be admitted to the $i$th buffer if no packets arrived for user $i$. Constraint \eqref{NRT_QoS} means that the queues of the NRT users have to be stable. Constraint \eqref{RT_QoS} indicates that the resources allocated to a RT user $i$ need to be such that the fraction of packets transmitted by the deadline are greater than the required QoS $q_i$. Constraint \eqref{P_avg} is an average power constraint on the BS transmission power. Finally constraint \eqref{Single_Tx_at_a_Time} guarantees that the sum of durations of transmission of all scheduled users does not exceed the slot duration $\Ts$. In this paper, we assume that the NRT user with the longest queue has enough packets, at each slot, to fit the whole slot duration which is a valid assumption in the heavy traffic regime. It will be clear that the generalization to the non-heavy traffic regime is possible by allowing multiple NRT users to be scheduled but this is omitted for brevity.

\section{Proposed Algorithm for On-Off Channels}
\label{Proposed_Algorithm}
We use the Lyapunov optimization technique \cite{li2011delay} to find and optimal algorithm that solves \eqref{Prob_DL}. We do this on four steps: i) We define, in Section \eqref{Prob_Decouple} a ``virtual queue'' associated with each average constraint in problem \eqref{Prob_DL}. This helps in decoupling the problem across time slots. ii) In Section \ref{Motivation_DL}, we define a Lyapunov function, its drift and a, per-slot, reward function. The latter is proportional to the objective of \eqref{Prob_DL}. iii) Based on the virtual queues and the Lyapunov function, we form an optimization problem, for each slot $k$, that minimizes the drift-minus-reward expression the solution of which is the proposed power allocation and scheduling algorithm. In Section \ref{Efficient_Solution}, we propose an efficient way to solve this problem optimally. iv) Finally, we show that this minimization guarantees reaching an optimal solution for \eqref{Prob_DL}, in Section \ref{Optimality_DL_Section}.

\subsection{Problem Decoupling Across Time Slots}
\label{Prob_Decouple}
We define a virtual queue associated with each RT user as follows
\begin{equation}
Y_i(k+1)=\lb \Yik + a_i(k)q_i-\Oneik\rb^+, \hspace{0.25in} i\in\sNR,
\label{DL_VQ}
\end{equation}
where $\Oneik\triangleq\mathds{1}\lb\muRik\rb$ with $\mathds{1}(\cdot)=1$ if its argument is non-zero and $\mathds{1}(\cdot)=0$ otherwise. For notational convenience we denote ${\bfY}(k)\parRdef{Y}$. $\Yik$ is a measure of how much constraint \eqref{RT_QoS} is violated for user $i$. We will later show a sufficient condition on $Y_i(k)$ for constraint \eqref{RT_QoS} to be satisfied. Hence, we say that the virtual queue $Y_i(k)$ is associated with constraint \eqref{RT_QoS}. Similarly, we define the virtual queue $X(k)$, associated with constraint \eqref{P_avg}, as
\begin{equation}
X(k+1)=\lb X(k) + \frac{\sum_{i\in\sN}\PRik \muRik}{\Ts}-\Pavg\rb^+.
\label{P_avg_VQ}
\end{equation}
To provide a sufficient condition on the virtual queues to satisfy the corresponding constraints, we use the following definition of \emph{mean rate stability} of queues \cite[Definition 1]{li2011delay} to state the lemma that follows.

\begin{Def}
\label{Mean_Rate_Def}
A random sequence $\{Y_i(k)\}_{k=0}^\infty$ is said to be mean rate stable if and only if $\limsup_{K\rightarrow\infty}\EE{Y_i(K)}/K=0$ holds.
\end{Def}

\begin{lma}
\label{Mean_Rate_Lemma}
If, for some $i\in\sNNR$, $\{Y_i(k)\}_{k=0}^\infty$ is mean rate stable, then constraint \eqref{RT_QoS} is satisfied for user $i$.
\end{lma}
\begin{proof}
Proof follows along the lines of Lemma 3 in \cite{li2011delay}.
\end{proof}
Lemma \ref{Mean_Rate_Lemma} shows that when the virtual queue $\Yik$ is mean rate stable, then constraint \eqref{RT_QoS} is satisfied for user $i\in\sNNR$. Similarly, if $\{X(k)\}_{k=0}^\infty$ is mean rate stable, then constraint \eqref{P_avg} is satisfied. Thus, our objective would be to devise an algorithm that guarantees the mean rate stability of $\Yik$ for all RT users as well as the mean rate stability for $X(k)$.

\subsection{Applying the Lyapunov Optimization}
\label{Motivation_DL}
The quadratic Lyapunov function is defined as
\begin{equation}
L_{\rm yap}\lb U(k)\rb\triangleq \frac{1}{2}\sum_{i\in\sNR}{Y_i^2(k)}+\frac{1}{2}\sum_{i\in\sNNR}{Q_i^2(k)}+\frac{1}{2}X^2(k),
\label{Lyapunov_Func}
\end{equation}
where $\bfU(k)\triangleq \lb\bfYk,\bfQk,X(k)\rb$, and the Lyapunov drift as $\Delta (k) \triangleq \E_{U(k)}[L_{k+1}\lb {\bf U}(k+1)\rb - L_{\rm yap}\lb \bfU(k)\rb]$ where $\EEU{x}\triangleq \EE{x\vert U(k)}$ is the conditional expectation of the random variable $x$ given $U(k)$. Squaring \eqref{Queues}, \eqref{DL_VQ} and \eqref{P_avg_VQ} taking the conditional expectation then summing over $i$, the drift becomes bounded by
\begin{equation}
\Delta(k)\leq C_1+\Psi(k),
\label{Drift_Bound}
\end{equation}
where
\begin{equation}
C_1\triangleq \frac{\underset{i\in\sNR}{\sum}\lb q_i^2+1\rb+\Pmax^2+\Pavg^2+\NNR\left[ L^2+\Ts^2\Rmax^2\right]}{2}
\end{equation}
and we use $\Rmax\triangleq\log\lb1+\Pmax\rb$, while
\begin{align}
\nonumber\Psi(k)\triangleq &\sum_{i\in\sNR}\EEU{\Yik\lb \lambda_i q_i-{\Oneik}\rb}\\
\nonumber+&X(k)\lb\sum_{i\in\sN}\frac{\EEU{\muRik\PRik}}{\Ts}-\Pavg\rb\\
+&\sum_{i\in\sNNR}\Qik\lb \EEU{Lr_i(k)-\muRik\Rik}\rb.
\label{Psi_k}
\end{align}
We define $\Bmax$ as an arbitrarily chosen positive control parameter that controls the performance of the algorithm. We shall discuss the tradeoff on choosing $\Bmax$ later on. Since $\EEU{Lr_i(k)}$ represents the average number of bits admitted to NRT user $i$'s buffer at slot $k$, we refer to $\Bmax \sum_{i\in\sNNR}\EEU{Lr_i(k)}$ as the ``reward term''. We subtract this term from both sides of \eqref{Drift_Bound}, then use \eqref{Psi_k} and rearrange to bound the drift-minus-reward term as
\begin{multline}
\Delta(k)-\Bmax \sum_{i\in\sNNR}\EEU{Lr_i(k)}\leq C_1-X(k)\Pavg+\\
\EEU{\sum_{i\in\sNR}\PsiRik}+\EEU{\sum_{i\in\sNNR}\PsiNRik\muRik}\\
+\EEU{\sum_{i\in\sNNR}\lb\Qik-\Bmax \rb Lr_i(k)}+\sum_{i\in\sNR}\Yik\lambda_i q_i,
\label{Drift_minus_Reward_Bound}
\end{multline}
where $\PsiRik$ and $\PsiNRik$ are given by
\begin{align}
&\PsiRik\triangleq \lb\Yik-\frac{L}{\Ts\Rik}X(k)\PRik\rb\Oneik, \hspace{0.05in} i\in\sNR,
\label{PsiRik}\\
&\PsiNRik\triangleq \Qik\Rik-\frac{X(k)\PRik}{\Ts}, \hspace{0.05in} i\in\sNNR,
\label{PsiNRik}
\end{align}
respectively, where we used \eqref{Num_Slots} in \eqref{PsiRik}. The proposed algorithm schedules the users, allocates their powers and controls the packet admission to minimize the right-hand-side of \eqref{Drift_minus_Reward_Bound} at each slot. Since the only term in right-hand-side of \eqref{Drift_minus_Reward_Bound} that is a function in $r_i(k)$ $\forall i\in\sNNR$ is the fourth term, we can decouple the admission control problem from the joint scheduling-and-power-allocation problem. Minimizing this term results in the following admission controller: set $r_i(k)=a_i(k)$ if $Q_i(k)<\Bmax$ and $0$ otherwise. Minimizing the remaining terms yields
\begin{equation}
\begin{array}{ll}
	&\underset{\bfPRk,\bfmuRk}{\text{maximize}}\sum_{i\in\SRk}\PsiRik + \sum_{i\in\sNNR}\PsiNRik\muRik\\
	&\text{subject to } \eqref{P_max}, \eqref{Single_Tx_at_a_Time} \text{ and } \eqref{muRik_Constr_NRT}.
\end{array}
\label{Max_Prob}
\end{equation}
This is a per-slot optimization problem the solution of which is an algorithm that minimizes the upper bound on the drift-minus-reward term defined in \eqref{Drift_minus_Reward_Bound}. Next we show how to solve this problem in an efficient way.
\subsection{Efficient Solution for the Per-Slot Problem}
\label{Efficient_Solution}
We first solve for the NRT variables then use its result to solve for the RT variables.
\subsubsection{NRT variables} To solve this problem optimally, we first find the optimal power-allocation-and-scheduling policy for the NRT users through the following lemma.

\begin{lma}
\label{NRT_Lemma_On_Off}
If an NRT user $i$ is scheduled to transmit any of its NRT data during the $k$th slot, then the optimum power level for this NRT with respect to (w.r.t.) problem \eqref{Max_Prob} is given by
\begin{equation}
\PNRik=\min\lb\lb\frac{\Ts Q_i(k)}{X(k)}-1\rb^+,\Pmax\rb.
\label{H2O_Pow_On_Off}
\end{equation}
Moreover, in the heavy traffic regime, the optimum NRT user to be scheduled, if any, w.r.t. problem \eqref{Max_Prob} is $\ist\triangleq\arg\max_{i\in\sNNR}\PsiNRikst$, where $\PsiNRikst$ comes by substituting \eqref{H2O_Pow_On_Off} in \eqref{PsiNRik}.
\end{lma}
\begin{proof}
We observe that, for any $i\in\sNNR$, the only term in \eqref{Max_Prob} that is a function in $\PRik$ is $\PsiNRik$. Differentiating \eqref{PsiNRik} w.r.t. $\PRik$ for all $i\in\sNNR$, equating the results to 0 and noting the minimum and maximum power constraints \eqref{P_max}, we get the water-filling power allocation formula \eqref{H2O_Pow_On_Off}. This completes the first part of the lemma.

To prove the second part, we substitute by \eqref{H2O_Pow_On_Off} in \eqref{PsiNRik} to get $\PsiNRikst$. We continue the proof by contradiction. Suppose that the optimal scheduled NRT set is given by $\SNRkst=\{\ist,j\}$ where $j\neq\ist$ and $\Psi_{\rm NR}^*(j,k)<\Psi_{\rm NR}^*(\ist,k)$. Thus, there exists some values $\alpha>0$ and $\beta>0$ such that the corresponding scheduler would be $\mu_{\ist}\lb k\rb=\alpha$ and $\mu_j\lb k\rb=\beta$, while $\mu_{l_k}(k)=0$ for all $l_k\notin\{\ist,j\}$. In other words, $\alpha$ seconds are assigned to $\ist$ and $\beta$ seconds assigned to $j$. However, if user $\ist$ has enough backlogged data, which happens in the heavy traffic regime, 
then we can increase its assigned duration to $\mu_{\ist}=\alpha+\beta$ and thus set $\mu_j(k)=0$, to get an increase in the objective of \eqref{Max_Prob} by $\beta \lb \Psi_{\rm NR}^*(\ist,k)-\Psi_{\rm NR}^*(j,k)\rb>0$ which contradicts with the optimality of $\SNRkst$ and completes the proof of the lemma.
\end{proof}

Lemma \ref{NRT_Lemma_On_Off} provides the optimal scheduling policy for the NRT users, at the $k$th slot, as well as the optimal power allocation w.r.t. problem \eqref{Max_Prob}. 
The lemma shows that if any of the NRT users is going to be scheduled in the $k$th slot, then only one of them is going to be scheduled. This means that the scheduling policy for the NRT users is
\begin{equation}
\muRik=\left\{
\begin{array}{lll}
	&\Ts-\sum_{i\in\SRstk}\muRik & i=\ist\\
	& 0 & \sNNR\backslash\{\ist\}
\end{array}
\right.
\label{mu_ist}
\end{equation}
which is a manipulation of \eqref{Single_Tx_at_a_Time}. Substituting \eqref{mu_ist} and $\PsiNRik$ in \eqref{Max_Prob}, the latter becomes
\begin{align}
	&\underset{\underset{\left[\muRik,\PRik\right]_{i\in\sNNR}}{\muRistk,}}{\text{maximize}}&\sum_{i\in\SRk}\PsiRik+\PsiNRistkst\muRistk
\label{Max_Prob_RT}\\
	&\nonumber\text{subject to } &\eqref{muRik_Constr_NRT}, \eqref{P_max} \text{ and }\muRistk=\Ts-\sum_{i\in\SRk}\frac{L}{\Rik},
\end{align}
which is simpler than \eqref{Max_Prob} since it is not a function in the NRT variables except $\muRistk$. Finding the optimal value of $\muRistk$ solves the NRT scheduling problem. We will first solve for $\muRik$ for all RT users then use \eqref{mu_ist} to find $\muRistk$.

\subsubsection{RT Variables} To find the scheduler of the RT users that is optimal w.r.t. problem \eqref{Max_Prob_RT}, we first solve for $\left[\PRik\right]_{i\in\sNR}$ given a fixed set $\SRk$, then we discuss the scheduling policy that solves for this set. To solve for $\left[\PRik\right]_{i\in\sNR}$, we present the following definition then present a theorem that discusses the optimum power allocation policy for the RT users.

\begin{Def}
\label{Lambert_Pow_Policy}
We define the Lambert power allocation policy for the RT users as
\begin{equation}
\PRik=\min\lb\frac{\frac{\Ts\Psi_{\rm NR}^*(\ist,k)}{X(k)}-1}{\LW{\left[\frac{\Psi_{\rm NR}^*(\ist,k)\Ts}{X(k)}-1\right]}}-1,\Pmax\rb,
\label{Lambert_Pow_On_Off}
\end{equation}
$i\in\SRk$, where $W_0(z)$ is the principle branch of the Lambert W function \cite{Lambert_W_Function} while $\PsiNRikst$ is defined in Lemma \ref{NRT_Lemma_On_Off}.
\end{Def}

\begin{thm}
\label{Pow_Alloc_Nec_Thm}
Given any set $\SRk$, if the Lambert power policy results in $\sum_{i\in\SRk}L/\log(1+\PRik)\leq\Ts$, then it is the optimum RT-users' power allocation policy given that $\SRk$ is the scheduling set at slot $k$. Otherwise, the optimum power allocation policy is given by
\begin{equation}
\PRik=\exp\lb{\frac{\underset{i\in\SRk}{\sum}L}{\Ts}}\rb-1,\hspace{0.25in} i\in\SRk.
\label{RT_NO_NRT_Pow}
\end{equation}
\end{thm}
\begin{proof}
We prove this theorem by applying the Lagrange optimization \cite[Ch. 5]{cvx_Boyd} technique to problem \eqref{Max_Prob_RT} then use the complementary slackness condition.

Since $\muRik\geq0$ for all $i\in\sNR$ (see \eqref{Num_Slots}), then we have the constraint $\muRistk\leq \Ts$ always holds from \eqref{mu_ist}. Thus we define the Lagrange multiplier $\phi$ to be the multiplier associated with the constraint $\muRistk\geq0$. The Lagrangian becomes
\begin{multline}
\Lag\triangleq \sum_{i\in\SRk}\PsiRik + \lb\PsiNRistkst+\phi\rb \times \\
\lb\Ts-\sum_{i\in\SRk}\frac{L}{\log\lb1+\PRik\gammaik\rb}\rb
\label{Lagrangian}
\end{multline}
Differentiating \eqref{Lagrangian} with respect to $\PRik$ and equating to $0$ gives
\begin{multline}
\log\lb 1+\PRik\gammaik\rb \frac{X(k)L}{\Ts}-\\\frac{\lb X(k)\PRik/\Ts+\PsiNRistkst+\phi\rb\gammaik}{1+\PRik\gammaik}=0.
\label{Deriv_Lagr}
\end{multline}
After some manipulations and denoting
\begin{equation}
\tilde{\phi}\triangleq\lb\PsiNRistkst+\phi\rb\Ts/X(k)
\label{phi_tilde}
\end{equation} we get
\begin{equation}
\log\lb1+\PRik\gammaik\rb=1+\frac{\tilde{\phi}\gammaik-1}{1+\PRik\gammaik}\triangleq 1+\tilde{P}.
\label{P_tilde}
\end{equation}
Thus we get $\tilde{P}e^{\tilde{P}}=\lb \tilde{\phi}\gammaik-1\rb e^{-1}$ which has two solutions in $\tilde{P}$ (see \cite{Lambert_W_Function}), one of them yields a negative value for $\PRik$. Hence, with the help of $W_0(\cdot)$, which is the inverse function of $xe^x$, we can write a unique solution for \eqref{Deriv_Lagr} as
\begin{equation}
\PRik=\frac{1}{\gammaik}\left[\frac{\tilde{\phi}\gammaik-1}{\LW{\left[\tilde{\phi}\gammaik-1\right]}}-1\right], \hspace{0.05in} i\in\SRk.
\label{Lambert_Pow_w_phi}
\end{equation}
To calculate \eqref{Lambert_Pow_w_phi}, we need to find the value of $\phi$ satisfying the complementary slackness condition $\phi\muRistk=0$. Hence we have one of the two following possibilities might yield the optimal solution: 1) setting $\phi=0$ and thus $\muRistk\geq0$, or 2) setting $\muRistk=0$ and thus $\phi\geq0$. If setting $\phi=0$ yields $\sum_{i\in\SRk}L/\log(1+\PRik)\leq\Ts$ then the Lambert power allocation policy in \eqref{Lambert_Pow_On_Off} is optimum since there exists no other non-negative value for $\phi$ that yields $\sum_{i\in\SRk}L/\log(1+\PRik)=\Ts$ while satisfying $\muRistk=0$ (to satisfy the complementary slackness). On the other hand, if setting $\phi=0$ yields $\sum_{i\in\SRk}L/\log(1+\PRik)>\Ts$, then $\phi$ cannot be $0$. Thus we have $\muRistk=0$, which means that the time slot will be allocated for RT users only. The corresponding value of $\phi$ should satisfy $\sum_{i\in\SRk}L/\log(1+\PRik)=\Ts$. From \eqref{Lambert_Pow_w_phi}, we observe that $\PRik=P_j(k)$ for all $i,j\in\SRk$ because $\gammaik=1$ for all $i\in\SRk$. Thus we have $L\vert\SRk\vert/\log(1+\PRik)=\Ts$. This yields the power allocation policy \eqref{RT_NO_NRT_Pow} and completes the proof.
\end{proof}
Theorem \ref{Pow_Alloc_Nec_Thm} gives closed-form expressions for the power function of the RT users given any scheduling set $\SRk$. To find the optimum scheduling set $\SRk$ that solves problem \eqref{Max_Prob_RT}, we present the following definition then mention a theorem that decreases the complexity of this search.
\begin{Def}
\label{Candidate_Sets}
At slot $k$, the set $\SRk$ is said to be a ``candidate'' set if and only if $\Yik\geq Y_j(k)$ for all $i\in\SRk$ and all $j\notin\SRk$. Otherwise it is called a ``non-candidate'' set.
\end{Def}
We note that the definition of candidate sets assumes that all RT users have $\gammaik=1$. If this assumption does not hold at some time slot $k$, then we eliminate the users with $\gammaik=0$ from the system for this time slot and consider only those with $\gammaik=1$.
\begin{thm}
\label{NRik_Lemma_On_Off}
The optimal RT set that solves \eqref{Max_Prob_RT} is one of the candidate sets.
\end{thm}
\begin{proof}
We prove this theorem by contradiction. Suppose that $\SRkst$ is the optimal set and that it is not a candidate set. That is, $\exists i\in\SRk$ and $j\notin\SRk$ such that $\Yik<Y_j(k)$. It is easy to show that the Lambert power policy results in the fact that $\PRik$ depends on $\vert\SRk\vert$ and not on $\SRk$ for any $i\in\SRk$ and any $\SRk$. Thus, replacing user $i$ with user $j$ results in having $P_j(k)=\PRik$ which means that $X(k)P_j(k)\mu_j(k)=X(k)\PRik\muRik$ holds. But since $\Yik<Y_j(k)$, swapping the two users increases the objective function of \eqref{Max_Prob_RT} and results in a candidate set. This contradicts with the fact that $\SRk$ is optimal while being non-candidate.
\end{proof}

Theorem \ref{NRik_Lemma_On_Off} says that there will be no scheduled RT users having a value of $Y_j(k)$ smaller than any of the unscheduled RT users. This theorem suggests an algorithm to reduce the complexity of scheduling the RT users from $O\lb2^{\NR}\rb$ to $O\lb \NR\rb$. This algorithm is to list the RT users in a descending order of their $Y_i(k)$. Without loss of generality, in the remaining of this paper, we will assume that $Y_1>Y_2\cdots >Y_{\NR}$. 

We now propose Algorithm \ref{Scheduling_Alg} which is the scheduling and power allocation algorithm for problem \eqref{Prob_DL}. Algorithm \ref{Scheduling_Alg} is executed at the beginning of the $k$th slot and, without loss of generality, it assumes: 1) all RT users in the system have received a packet at the beginning of the $k$th slot, 2) all users in the system have an ``on'' channel. If, at some slot, any of these assumptions does not hold for some users, these users are eliminated from the system for this slot. That is, they will not be scheduled. In addition, we assume heavy traffic regime, thus the NRT user with the longest queue has enough data to fill the entire time slot. We define the set $\SRT$ to be the set of all candidate sets.

\begin{algorithm}
\caption{Scheduling and Power Allocation Algorithm}
\begin{algorithmic}[1]
\label{Scheduling_Alg}
\STATE Define the auxiliary functions $\Psiaux(\cdot):\SRT\rightarrow \mathbb{R}_+$ and $\Paux(\cdot,\cdot):\SRT\times\sNR\rightarrow\mathbb{R}_+$.
\STATE Initialize $\Paux(\sS,i)=0$ for all $\sS\in\SRT$ and all $i\in\sNR$.
\STATE Sort the RT users in a descending order of $\Yik$. Without loss of generality, assume that $Y_1>Y_2\cdots>Y_{\NR}$.
\STATE Find the user $\ist$ with longest queue $\Qik$ and set $\SRk$ to be an empty set.
\WHILE{$i\leq\NR$}
\STATE $\SRk=\SRk\cup\{i\}$ and set the power according to \eqref{Lambert_Pow_On_Off} $\forall i\in\SRk$.
\STATE Calculate $\muRik$ and $\muRistk$ according to \eqref{Num_Slots} and \eqref{mu_ist}, respectively.
\IF {$\muRistk<0$}
\STATE Set $\muRik=0$ for all $i\in\sNNR$ and set the power allocation for all $i\in\SRk$ according to \eqref{RT_NO_NRT_Pow} and recalculate $\muRik$ according to \eqref{Num_Slots}.
\ENDIF
\STATE Set $\Psiaux(\SRk)=\sum_{i\in\SRk} \lb Y_i(k)-X_i(k)\muRik\rb + \PsiNRistkst\muRistk$.
\STATE Set $\Paux(\SRk,i)=\PRik$, $\forall i\in\SRk$.
\STATE $i\leftarrow i+1$.
\ENDWHILE
\STATE Set the optimum scheduling set $\SRstk=\arg\max_{\SRk}\Psiaux(\SRk)$.
\STATE Set $\PRikst=\Paux\lb\SRstk,i\rb$ for all $i\in\sNR$, and set the NRT scheduler according to \eqref{mu_ist}.
\STATE For each $i\in\sNNR$, set $r_i(k)=a_i(k)$ if $Q_i(k)<\Bmax$ and $0$ otherwise.
\STATE Update equations \eqref{Queues}, \eqref{DL_VQ} and \eqref{P_avg_VQ} at the end of the $k$th slot.
\end{algorithmic}
\end{algorithm}

\subsection{Optimality of Proposed Algorithm}
\label{Optimality_DL_Section}
We first define $\bRopt$ to be the throughput of NRT user $i$ under the optimal algorithm that solves \eqref{Prob_DL}. We define this algorithm to be the one that sets, at each time slot $k$, the variables $\PRik$, $\muRik$, $\Oneik$ and $\Rik$ to the values $\PRopt$, $\muRopt$, $\Oneopt$ and $\Ropt$, respectively, where the latter values satisfy
\begin{align}
\limsup_{K\rightarrow\infty}\frac{1}{K}\sum_{k=0}^{K-1}\EE{\Oneopt}\geq &\lambda_iq_i, \hspace{0.2in} \forall i\in\sNR,
\label{RT_QoS_Opt}\\
\limsup_{K\rightarrow\infty}\frac{1}{K}\sum_{k=0}^{K-1}\sum_{i\in\sN}\EE{\frac{\muRopt\PRopt}{\Ts}}\leq &\Pavg,
\label{P_avg_Opt}\\
\limsup_{K\rightarrow\infty}\frac{1}{K}\sum_{k=0}^{K-1}\EE{\frac{\muRopt\Ropt}{L}}=&\bRopt, \hspace{0.2in} \forall i\in\sNNR,
\label{NRT_QoS_Opt}
\end{align}
where $\bRopt$ is the optimal rate for user $i\in\sNNR$ with respect to solving \eqref{Prob_DL}. The following theorem gives a bound on the performance of Algorithm \ref{Scheduling_Alg} compared to the optimal algorithm that has a genie-aided knowledge of $\bRopt$ which, we show that, due to this knowledge it can solve the problem optimally.
\begin{thm}
\label{Optimality_Thm}
For the on-off channel model, if problem \eqref{Prob_DL} is feasible, then for any $\Bmax >0$ Algorithm \ref{Scheduling_Alg} results in satisfying all constraints in \eqref{Prob_DL} and achieves an average rate satisfying
\begin{equation}
\sum_{i\in\sNNR} \bmuRi\geq \sum_{i\in\sNNR}{\bRopt} - \frac{C_1}{L\Bmax }.
\label{Optimality_Eq}
\end{equation}
\end{thm}

\begin{proof}
See Appendix \ref{App_Optimality_Thm_Proof}
\end{proof}

Theorem \ref{Optimality_Thm} says that Algorithm \ref{Scheduling_Alg} yields an objective function \eqref{Prob_DL} that is arbitrary close to the performance of the optimal algorithm that solves \eqref{Prob_DL}.

\section{Extensions to Continuous Fading Channels}
\label{Continuous_Fading}
In the case of continuous fading, i.e. $\gammaik\in [0,\gammamax]$ where $\gammamax<\infty$ is the maximum channel gain that $\gammaik$ can take, we expect the power allocation to depend on the channel gain. An algorithm that solves this case is a generalization of Algorithm \ref{Scheduling_Alg} that assumes $\gammaik\in\{0,1\}$. However, as will be demonstrated later, the scheduling algorithm of the RT users has a higher complexity order than the special case of on-off channel gains.

We adopt the same model as in Section \ref{Model} except that we allow $\gammaik$ to take any value in the interval $[0,\gammamax]$, for all $i\in\sN$. The transmission rate for this case is still given by \eqref{Rate_On_Off}, and the optimization problem is the same as \eqref{Prob_DL} with a new assumption for $\gammaik$.

\subsection{Derivation of the Algorithm}
\label{Derivation_Cont}
Algorithm \ref{Scheduling_Alg_Cont} is based on the same Lyapunov optimization procedure as explained in Section \ref{Proposed_Algorithm}. Following this procedure, we reach optimization problem \eqref{Max_Prob_RT} with the new definition of $\gammaik$. We now present the solution for the NRT users followed by that of the RT users.

\begin{lma}
\label{NRT_Lemma_Cont}
If user $i\in\sNNR$ is scheduled to transmit any of its NRT data during the $k$th slot, then the optimum power level for this NRT w.r.t. problem \eqref{Max_Prob_RT} in the continuous fading case is given by
\begin{equation}
\PNRik=\min\lb\lb\frac{Q_i(k)}{X(k)}-\frac{1}{\gammaik}\rb^+,\Pmax\rb.
\label{H2O_Pow_Cont}
\end{equation}
Moreover, in the heavy traffic regime, the scheduled NRT user, if any, that optimally solves problem \eqref{Prob_DL} is given by
\begin{equation}
\ist=\arg\max_{i\in\sNNR}\PsiNRikst,
\label{ist}
\end{equation}
with ties broken randomly uniformly, while $\PsiNRikst$ comes by substituting \eqref{H2O_Pow_Cont} in \eqref{PsiNRik}.
\end{lma}
\begin{proof}
The proof is similar to that of Lemma \ref{NRT_Lemma_On_Off} and is omitted for brevity.
\end{proof}

Lemma \ref{NRT_Lemma_Cont} presents the optimal power and scheduling policy for the NRT users. To solve for the RT users, we assume a fixed subset $\SRk\subseteq\sNR$ of RT users to be scheduled during the $k$th slot and find the power allocation of these users. Consequently, the optimum set $\SRkst$ is the one that maximizes \eqref{Max_Prob_RT}. In Section \ref{Proposed_Algorithm_Cont}, we present an algorithm that finds this optimum set as well as discussing the complexity of this algorithm.

Assuming that the users in the set $\SRk$ are scheduled at the $k$th slot, the problem is to find the transmission power levels for all the users in this set. We answer this question in the following theorem.

\begin{thm}
\label{PRik_Lemma_Cont}
In the continuous-fading channel model, given some non-empty set $\SRk$, the power allocation policy
\begin{equation}
\PRik=\min\lb\frac{1}{\gammaik}\left[\frac{\tilde{\phi}\gammaik-1}{\LW{\left[\tilde{\phi}\gammaik-1\right]}}-1\right],\Pmax\rb,
\label{Lambert_Pow}
\end{equation}
$i\in\SRk$, with $\tilde{\phi}\triangleq\lb\PsiNRistkst+\phi\rb\Ts/X(k)$ and $\PsiNRikst$ defined in Lemma \ref{NRT_Lemma_Cont}, is optimal w.r.t. \eqref{Max_Prob_RT} when $\phi$ is set to a non-negative value that satisfies \eqref{Single_Tx_at_a_Time}.
\end{thm}
\begin{proofsketch}
The proof is similar to that of Theorem \ref{Pow_Alloc_Nec_Thm}. The only difference is that we have to obtain the optimum value of $\phi$ satisfying \eqref{Single_Tx_at_a_Time}. We note that instead of finding $\phi>0$ using a 1-dimensional grid search, we can use the bisection method \cite[Ch.9]{Numerical_Recipes_Ch9} which requires the monotonicity of the left-hand-side of \eqref{Single_Tx_at_a_Time}, a fact that can be shown easily by showing that the derivative, of this left-hand-side, with respect to $\phi$ is always negative. Moreover, since the bisection algorithm needs a bracketing interval, it can be easily shown that the optimum $\phi$ satisfies $\phi\leq \phimax\triangleq-\PsiNRistkst+\exp(L\vert\SRk\vert/\Ts)L\vert\SRk\vert X(k)\Pmax/(\exp(L\vert\SRk\vert/\Ts)-1)$.
\end{proofsketch}

It is clear that the Lambert power policy in \eqref{Lambert_Pow} has a different structure than the water-filling policy in \eqref{H2O_Pow_Cont}. The reason is because the former is for the RT users while the later is for the NRT users. We plot the two policies in Fig. \ref{Plot_H2O_vs_Lambert} with $L=1$, $\Ts=1$, $\Pmax=20$ while $\Qik/X(k)=15$. The Lambert policy is plotted assuming a single RT user is scheduled at slot $k$ while the water-filling policy is plotted assuming a single NRT user is scheduled at slot $k$. We note that when a RT user $i$ is the only scheduled user, \eqref{Lambert_Pow} is equivalent to
\begin{equation}
\PRik=\min\lb\frac{e^{L/\Ts}-1}{\gammaik},\Pmax\rb,
\label{Lambert_Pow_Single}
\end{equation}
We contrast the fact that, while the water-filling is an increasing function in the channel gain, the Lambert is a decreasing function in the channel gain. This is because the RT user has a single packet of a fixed length to be transmitted. If the channel gain increases, then the power decreases to keep the same transmission rate resulting in the same transmission duration of one slot. This result holds when multiple RT users are scheduled as well as demonstrated in the following theorem.
\begin{figure}
	\centering
		\includegraphics[width=1\columnwidth]{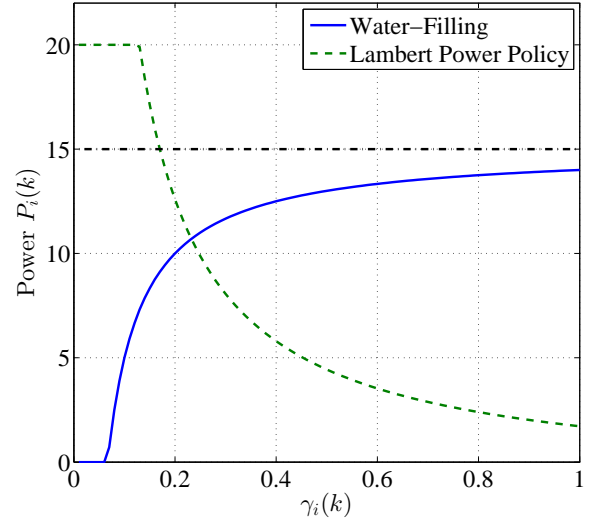}
		\caption{The Lambert power policy decreases with the channel gain, while the water-filling policy increases with the gain.}
	\label{Plot_H2O_vs_Lambert}
\end{figure}

\begin{thm}
\label{Thm_Lambert_Monotonicity}
Let $\SRk$ be some scheduling RT set at slot $k$. The power $\PRik$ given by \eqref{Lambert_Pow} is monotonically decreasing in $\gammaik$ $\forall i\in \SRk$.
\end{thm}
\begin{proofsketch}
Proof follows by differentiating \eqref{Lambert_Pow} with respect to $\gammaik$ for some user $i$, while having $\phi$ satisfying \eqref{Single_Tx_at_a_Time}, and showing that the resulting derivative is always non-positive for $\gammaik\geq 0$.
\end{proofsketch}

The optimum scheduling algorithm for the RT users is to find, among all subsets of the set $\sNR$, the set that gives the highest objective function of \eqref{Max_Prob_RT}.

\subsection{Proposed Algorithm and Proof of Optimality}
\label{Proposed_Algorithm_Cont}
The exhaustive approach to the scheduling problem is to evaluate the objective function of \eqref{Max_Prob_RT} for all $2^{\NR}$ possible sets and choose the set that gives the highest objective function. This may be not practical when the number of RT users is large. Observing the approach in the special on-off case and inspired by Theorem \ref{NRik_Lemma_On_Off} that reduces the search space, we provide here a similar approach. We first provide the following definition which is analogous to Theorem \ref{NRik_Lemma_On_Off}.

\begin{thm}
\label{NRik_Lemma_Cont}
At slot $k$, for any set $\SRk$, if there exists some $i\notin\SRk$ and some $j\in\SRk$ such that $\Yik>Y_j(k)$ and $\gammaik>\gamma_j(k)$, then $\SRk$ cannot be an optimal RT set, with respect to problem \eqref{Max_Prob_RT}, for the continuous channel model.
\end{thm}
\begin{proofsketch}
The proof is carried out by contradiction. We can show that if $\Yik>Y_j(k)$ and $\gammaik>\gamma_j(k)$ for some $i\notin\SRk$ and some $j\in\SRk$, then we could form another set $\script{S}'(k)$ by swapping users $i$ and $j$ and thus increase the objective function of \eqref{Max_Prob_RT}.
\end{proofsketch}
This theorem provides a sufficient condition for non-optimality. In other words, we can make use of this theorem to restrict our search algorithm to the sets that do not satisfy this property. Before presenting the proposed algorithm, we define the set $\SRT$ as the set of all possible subsets of the set $\sNR$.
\begin{algorithm}
\caption{Lambert-Strict Algorithm}
\begin{algorithmic}[1]
\label{Scheduling_Alg_Cont}
\STATE Define the auxiliary functions $\Psiaux(\cdot):\SRT\rightarrow \mathbb{R}_+$ and $\Paux(\cdot,\cdot):\SRT\times\sNR\rightarrow\mathbb{R}_+$.
\STATE Initialize $\Paux(\sS,i)=0$ for all $\sS\in\SRT$ and all $i\in\sNR$.
\STATE Find the user $\ist$ given in \eqref{ist} and calculate its power given by \eqref{H2O_Pow_Cont}.
\FOR {$\sS\in \SRT$}
\IF{$\exists$ some $i\notin\sS$ and some $j\in\sS$ such that $\Yik>Y_j(k)$ and $\gammaik>\gamma_j(k)$}
\STATE Set $\Psiaux(\sS)=-\infty$.
\STATE Skip this iteration and go to step 4 to continue with the next set in $\SRT$.
\ENDIF
\STATE $\phi\leftarrow\phimax+\Delta\phi$
\WHILE{$\phi\muRik\neq 0$}
\STATE $\phi\leftarrow\phi-\Delta\phi$
\STATE Calculate $\PRik$ given by \eqref{Lambert_Pow} for all $i\in\sS$ and set $\muRistk= \Ts-\sum_{i\in\sS}\muRik$.
\ENDWHILE
\STATE Set $\Psiaux(\sS)=\sum_{i\in\sS} \lb Y_i(k)-X_i(k)\muRik\rb + \PsiNRistkst\muRistk$ and $\Paux(\sS,i)=\PRik, \hspace{0.25in}i\in\sS$.
\STATE Set $\Psiaux(\sS)=\sum_{i\in\sS} \lb Y_i(k)-X_i(k)\muRik\rb + \PsiNRistkst\muRistk$.
\STATE Set $\Paux(\sS,i)=\PRik$, $\forall i\in\sS$.
\STATE $i\leftarrow i+1$.
\ENDFOR
\STATE Set the optimum scheduling set $\SRstk=\arg\max_\sS\Psiaux(\sS)$.
\STATE Set $\PRikst=\Paux\lb\SRstk,i\rb$ for all $i\in\sNR$, and set the NRT scheduler according to \eqref{mu_ist}.
\STATE For each $i\in\sNNR$, set $r_i(k)=a_i(k)$ if $Q_i(k)<\Bmax$ and $0$ otherwise.
\STATE Update equations \eqref{Queues}, \eqref{DL_VQ} and \eqref{P_avg_VQ} at the end of the $k$th slot.
\end{algorithmic}
\end{algorithm}

\begin{thm}
\label{Optimality_Thm_Cont}
For the continuous channel model, if problem \ref{Prob_DL} is feasible, then for any $\Bmax >0$ and any $\epsilon\in(0,1]$ there exists some finite constant $C_2$ such that Algorithm \ref{Scheduling_Alg_Cont} satisfies all constraints in \eqref{Prob_DL} and achieves an average sum throughput satisfying
\begin{equation}
\sum_{i\in\sNNR} \bmuRi\geq \sum_{i\in\sNNR}{\bmuRi^*} - \frac{C_2}{L\Bmax },
\label{Optimality_Eq_Cont}
\end{equation}
where $\bmuRi^*$ is the optimal rate for user $i$ w.r.t. \eqref{Prob_DL}.
\end{thm}

\begin{proof}
The proof is similar to that of Theorem \ref{Optimality_Thm} and $C_2$ is defined as $C_1$ but with $\Rmax\triangleq\log\lb1+\Pmax\gammamax\rb$. We omit the proof for brevity.
\end{proof}


Due to the problem being a combinatorial problem with a huge amount of possibilities, we could not reach a closed-form expression for the complexity order of this algorithm. However, simulations will show its complexity improvement over the exhaustive search algorithm.

\subsection{Extensions to Packets with Different Lengths}
\label{General_L}
Let $L_i(k)$ be the length of user $i$'s packet at slot $k$, $i\in \sN$. It can be easily shown that the power-allocation-and-scheduling for the NRT users and the power allocation for the RT users, namely Lemma \ref{NRT_Lemma_Cont} and Theorem \ref{PRik_Lemma_Cont} do not change. That is, replacing $L$ with $L_i(k)$ in \eqref{RT_NO_NRT_Pow} yields an optimal solution as well. However, there are two possible extensions for the scheduling algorithm for the RT users, namely Algorithm \ref{Scheduling_Alg_Cont}. We discuss them next.

\subsubsection{Homogeneous RT Users}This is where all packets of all RT users have the same lengths at slot $k$ but they change (randomly and independently) from a slot to the other. That is, $L_i(k)=L_j(k)$ for all $i$ and $j$, but $L_i(k_1)$ and $L_i(k_2)$ need not be the same for $k_1\neq k_2$. This could be the case if all RT users are streaming the same information from the same server, or if their packet lengths change from a slot to the other but are highly correlated across users in the sense that $L_i(k)=L_j(k)$ is a valid approximation. In this case, Algorithm \ref{Scheduling_Alg_Cont} is still optimal since it solves problem \eqref{Max_Prob_RT} which is a per-slot optimization problem, namely, it is not affected with the packet lengths at preceding and succeeding time slots.

\subsubsection{Heterogeneous RT Users}This is where the packet length changes significantly from a user to the other in addition to its change (randomly and independently) from a slot to the other. In this case, the scheduling algorithm of the RT users proposed in Algorithm \ref{Scheduling_Alg_Cont} is suboptimal. In order for the algorithm to be optimal, Steps 5 through 8 of the algorithm need to be removed. That is, the algorithm goes over all subsets of the set $\sNR$. The complexity of the optimal algorithm is exponential in the number of RT users. However, suboptimal algorithms could still be developed. One example is to modify Step 3 of Algorithm \ref{Scheduling_Alg} by sorting the users according to a decreasing order of $Y_i(k)\gammaik/L_i(k)$. Consequently, this yields an algorithm of a linear complexity in $\NR$. The sorting according to $Y_i(k)\gammaik/L_i(k)$ stems from the fact that RT users with higher $\Yik$ and $\gammaik$ and lower $L_i(k)$ should be more favored to be scheduled.

\section{Capacity Region}
\label{Capacity_Region}
In Section \ref{Continuous_Fading}, Algorithm \ref{Scheduling_Alg_Cont} is shown to maximize the NRT sum-throughput subject to the system constraints. In this section we want to study the stability of the system. Specifically, we are interested to answer the following two questions:
\begin{enumerate}
	\item What is the capacity region of the system under the continuous fading model?
	\item What scheduling and power-allocation algorithms can achieve this capacity region?
\end{enumerate}

Studying the system's capacity region means that we need to find all arrival rate vectors $\lambdaNR$ under which the NRT users' queues are stable (i.e. have a stationary distribution). This needs to be studied assuming that all arriving packets are admitted to their respective buffers. Hence we first eliminate the admission controller $\bfrk$ by replacing the queue equation \eqref{Queues} with
\begin{equation}
Q_i(k+1)=\lb \Qik + La_i(k)-\muRik\Rik\rb^+.
\label{Queues_No_Admis}
\end{equation}
More formally, the first question now becomes: what is the closure of all admissible arrival rate vectors? An admissible arrival rate vector is defined next.
\begin{Def}
An arrival rate vector $\lambdaNR\triangleq [\lambda_i]_{i\in\sNNR}$ is said to be admissible if there exists a power-allocation and scheduling algorithm under which constraints \eqref{NRT_QoS} and \eqref{RT_QoS} are satisfied given the power and scheduling constraints \eqref{P_avg}-\eqref{muRik_Constr_NRT}.
\end{Def}
For simplicity we henceforth 
assume that the channel gain $\gammaik\in\sM$ where $\sM$ is a discrete finite set, the elements of which are in the range $[0,\gammamax]$. With a slight abuse in notation, we define $\gammaikm\triangleq\gammaik$ to be the gain of user $i$ when the channel is in fading state $\bfm\parmdef{\gamma}\in\sM^N$ during slot $k$. We also define $\muRikm$ and $\PRikm$ to be, respectively, the duration and power allocated to user $i\in\sN$ when the channel is in fading state $\bfm\parmdef{\gamma}\in\sM^N$ during slot $k$, and $\pim$ to be the probability of occurrence of fading state $\bfm$. We now mention the following definition then state Theorem \ref{Capacity_Thm} that answers the first question.

\begin{Def}
\label{Rate_Region}
An arrival rate vector $\lambdaNR$ is said to belong to the ``Lambert Region'' $\sR$ if and only if there exists a sequence of time duration vectors $\left\{\bfmuRk\right\}$ and a power allocation policy $\left\{\bfPRk\right\}$ that make $\lambdaNR$ satisfy
\begin{equation}
\lambda_i=\frac{1}{L}\sum_{\bfm\in\sM^N}\muRikm\log\lb1+\PRikm\gammaikm\rb \pim ,
\label{NRT_Arr_Rate}
\end{equation}
$i\in\sNNR$, while having $\left\{\bfmuRk\right\}$ and $\left\{\bfPRk\right\}$ satisfy
\begin{align}
&q_i\lambda_i\leq\sum_{\bfm\in\sM^N}\muRikm\log\lb1+\PRikm\gammaikm\rb , \hspace{0.05in} i\in\sNR,
\label{RT_Arr_Rate}\\
&\sum_{i\in\sN}\muRikm\leq \Ts, \hspace{0.25in} \forall k\geq 1,\bfm\in\sM^N,
\label{Single_Tx_Constr}\\
&\underset{K\rightarrow\infty}{\limsup}\frac{1}{K}\sum_{k=1}^K\sum_{i\in\sN}\sum_{m\in\sM}\muRikm\PRikm\leq \Pavg,
\label{P_avg_Capacity}\\
&\muRikm\geq 0 , \hspace{0.25in} i\in\sN, \forall k\geq 1,\bfm\in\sM^N,\\
&\PRikm\geq 0 , \hspace{0.25in} i\in\sN, \forall k\geq 1,\bfm\in\sM^N.
\label{PRikm_Positive}
\end{align}
\end{Def}

\begin{thm}
\label{Capacity_Thm}
If $\lambdaNR(1+\epsilon)\in\sR$ then Algorithm \ref{Scheduling_Alg_Cont} satisfies \eqref{NRT_QoS}-\eqref{muRik_Constr_NRT}. Otherwise, then problem \eqref{Prob_DL} is infeasible.
\end{thm}
\begin{proof}
See Appendix \ref{App_Capacity_Thm_Proof}
\end{proof}

Theorem \ref{Capacity_Thm} says that $\sR$ is in fact the system's capacity region. This answers the first question. Moreover, the second question is answered in the proof, as shown in Appendix \ref{App_Capacity_Thm_Proof}. In the proof, we show that with a simple modification to Algorithm \ref{Scheduling_Alg_Cont} we can achieve this capacity region. The modification is by setting $r_i(k)=a_i(k)$ for all $i\in\sNNR$.

\section{Simulation Results}
\label{Results}
We simulate the system for the on-off channel model as well as the continuous channel model. For both models, we assume that all channels are statistically homogeneous, i.e. $\bgamma_i=\bgamma$ for all $i\in\sN$ where $\bgamma$ is a fixed constant. Moreover, all RT users have homogeneous delivery ratio requirements, thus $q_i=q$ for all $i\in\sNR$ for some parameter $q$. All parameter values used in the simulations are: $L=1$ bits, $\Pmax=20$ and $\bgamma_i=1$.

We compare the throughput of the RT users, which is the objective of problem \eqref{Prob_DL}, to that of a simple power allocation and scheduling algorithm that we call ``FixedP'' algorithm. In the FixedP algorithm, all scheduled users transmit with the maximum power, i.e. $\PRik=P_{\rm max}$ for all $i\in\sN$ and all $k\geq1$, while the scheduling policy is to flip a biased coin and choose to schedule either the NRT users or the RT users. The coin is set to schedule the RT users with probability $q$ (the delivery ratio requirement for all users), at which case the RT users are sorted according to $Y_i(k)$ and scheduled one by one until the current slot ends. On the other hand, when the coin chooses the NRT users, the FixedP policy assigns the entire time slot to the NRT user with the longest queue.


\subsection{On-Off Channel Model}
We assume that we have $N=20$ users that is split equally between the RT and NRT users, i.e. $\NR=\NNR=20$. Fig. \ref{Algorithm1_vs_FixedP_P2_10} shows a substantial increase in the average rate of the proposed algorithm over the FixedP algorithm with over $200\%$ at low $\Pavg$ values and $60\%$ at high $\Pavg$ values. We simulated the system with $\Bmax=10^4$, $\Ts=1$ and $q=0.3$.

\begin{figure}%
\centering
\includegraphics[width=1.1\columnwidth]{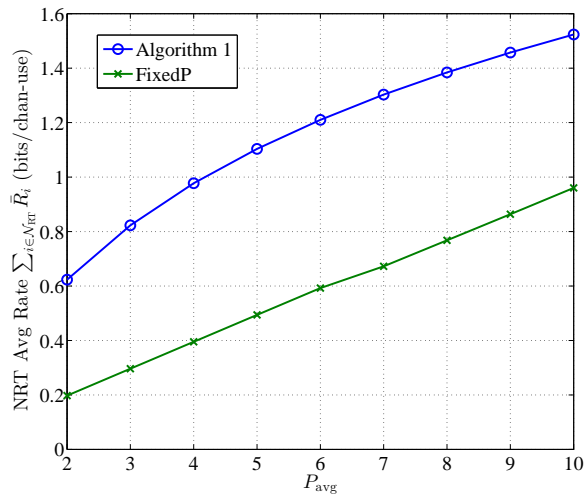}%
\caption{Sum of average throughput for all NRT users. The FixedP algorithm assigns a fixed power to all users set at $\Pmax$.}%
\label{Algorithm1_vs_FixedP_P2_10}%
\end{figure}

In Fig. \ref{Algorithm1_vs_FixedP_Qp1_8}, the sum of average NRT users' throughput is plotted while keeping $\Pavg=10$ but changing $q$. We can see that the FixedP algorithm results in a large degradation in the throughput compared to Algorithm \ref{Scheduling_Alg} which allocates the power and schedules the users optimally with respect to \eqref{Prob_DL}. The decrease in the throughput observed in both curves of Fig. \ref{Algorithm1_vs_FixedP_Qp1_8} is due to the increase in the parameter $q$. This increase makes constraint \eqref{RT_QoS} more stringent and thus decreases the feasible region decreasing the throughput.

\begin{figure}%
\centering
\includegraphics[width=1.1\columnwidth]{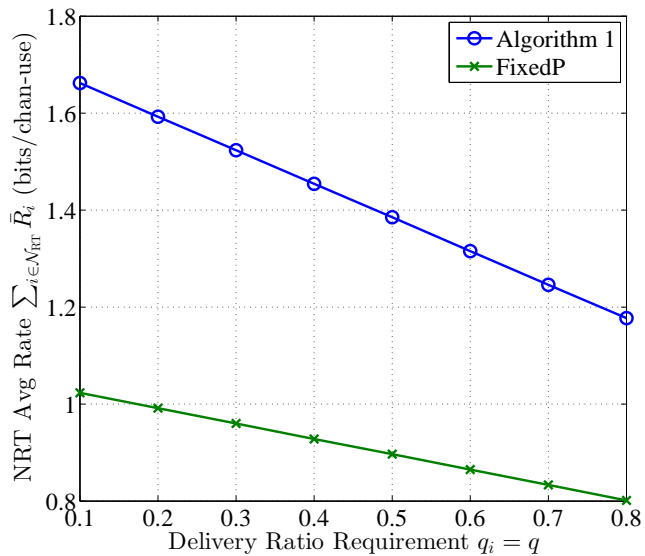}%
\caption{As $q$ increases, the RT users are assigned the channel more frequently. This comes at the expense of the NRT's throughput. However, the proposed algorithm outperforms the FixedP algorithm.}%
\label{Algorithm1_vs_FixedP_Qp1_8}
\end{figure}

In Fig. \ref{Algorithm1_vs_FixedP_N13_20} we show the effect of increasing the number of users on the system's throughput. As the number of users increase, more RT users have to be scheduled. This comes at the expense of the time allocated to the NRT users thus decreasing the throughput for the two plotted algorithms.
\begin{figure}%
\centering
\includegraphics[width=1\columnwidth]{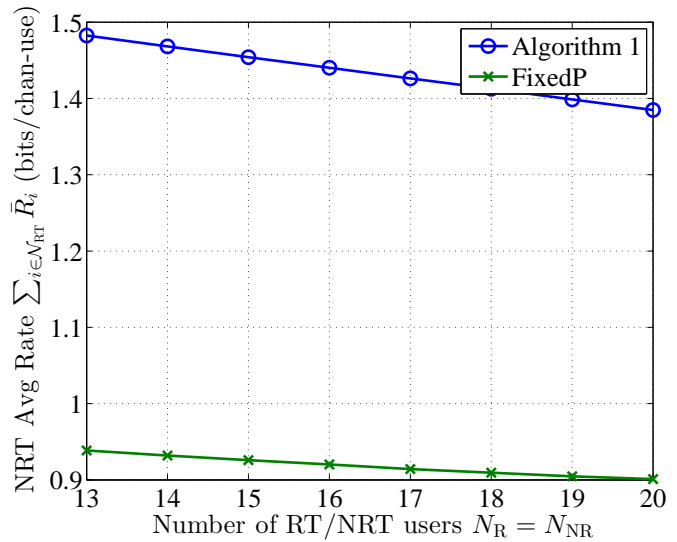}%
\caption{As $N$ increases, the RT users are allocated the channel more at the expense of the NRT users' throughput.}%
\label{Algorithm1_vs_FixedP_N13_20}%
\end{figure}

\subsection{Continuous Channel Fading Model}
In this simulation setup, we assume the channels are fading according to a Rayleigh fading model with avg power gain of $\bgamma=1$. 
%
In Fig. \ref{Cont_Opt_vs_LambStrict_Complexity}, we plot the complexity of the Lambert-Strict algorithm as well as the exhaustive search algorithm with exponential complexity versus the number of users $\NR$. The complexity is measured in terms of the average number of iterations, per-slot, where we have to evaluate the objective function of \eqref{Max_Prob_RT}. 
Since this complexity changes from a slot to the other, we plot the average of this complexity. As the number of users increases, the Lambert-Strict algorithm has an average complexity close to linear. However, the there is no sacrifice in the throughput of the NRT users. This is shown in Fig. \ref{Cont_Opt_vs_LambStrict_Throughput}. The reason stems from the optimality of the Lambert-Strict algorithm that does not eliminate any RT users from scheduling unless it is a suboptimal user. We note that we simulated this system with $\Bmax=100$, $\Ts=5$, $\Pavg=10$ and $q=0.9$.
\begin{figure}%
\centering
\includegraphics[width=1\columnwidth]{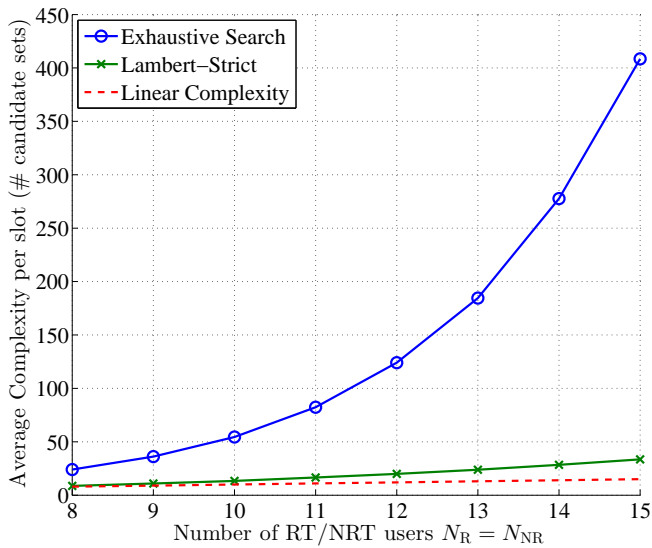}%
\caption{As the number of NRT users in the system increase the complexity increases exponentially for exhaustive search and nearly linear for the Lambert-Strict algorithm.}%
\label{Cont_Opt_vs_LambStrict_Complexity}%
\end{figure}

\begin{figure}%
\centering
\includegraphics[width=1.05\columnwidth]{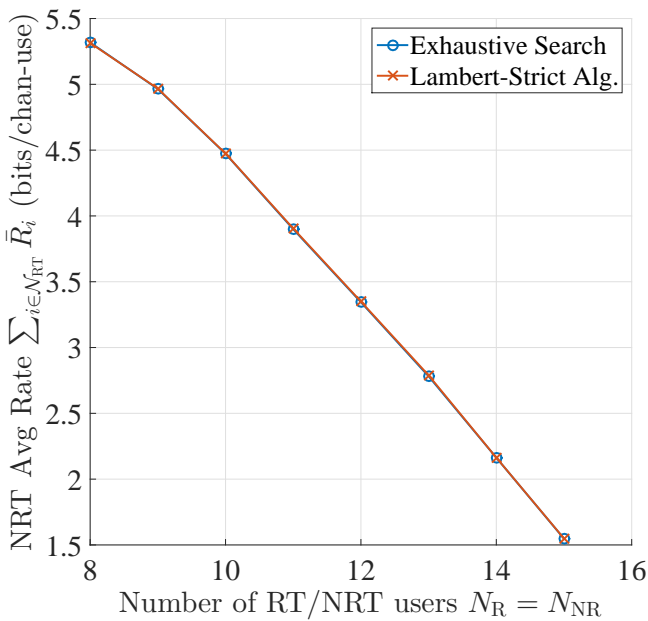}%
\caption{The Lambert-Strict Algorithm yields the same throughput as the exhaustive search algorithm but with a lower average complexity.}%
\label{Cont_Opt_vs_LambStrict_Throughput}%
\end{figure}

\section{Conclusions}
\label{Conclusion}
We discussed the problem of throughput maximization in downlink cellular systems in the presence of RT and NRT users. We formulated the problem as a joint power-allocation-and-scheduling problem. Using the Lyapunov optimization theory, we presented two algorithms to optimally solve the throughput maximization problem. The first algorithm is for the on-off channel fading model while the second is for the continuous channel fading model. The power allocations for both algorithms are in closed-form expressions for the RT as well as the NRT users. We showed that the NRT power allocation is water-filling-like which is monotonically increasing in the channel gain. On the other hand, the RT power allocation has a totally different structure that we call the ``Lambert Power Allocation''. It is found that the latter is a decreasing function in the channel gain.

The two algorithms differ in the complexity of the adopted scheduling policies. The first algorithm has a linear complexity while the second is shown, through simulations, to have a close-to-linear complexity. We presented the capacity region of the problem and showed that the proposed algorithms achieve this region.

\appendices

\section{Proof of Theorem \ref{Optimality_Thm}}
\label{App_Optimality_Thm_Proof}
\begin{proof}
We divide the proof into two parts. First, we show that the virtual queues are mean rate stable. This proves that constraints \eqref{RT_QoS} and \eqref{P_avg} are satisfied. Second, through the Lyapunov optimization technique we show that the drift-minus-reward term is within a constant gap from the performance of the optimal, genie-aided algorithm \cite{georgiadis2006resource,urgaonkar2011optimal}.

\subsubsection{Mean Rate Stability} 
According to \eqref{Max_Prob}, Algorithm \ref{Scheduling_Alg} minimizes $\Psi(k)$ where the minimization is taken over all possible scheduling and power allocation algorithms including the optimal algorithm that solves \eqref{Prob_DL}. We define $\Psi^*(k)\triangleq\min\Psi(k)$. Thus we can write $\Psi^*(k)\leq \Psiopt$ 
where $\Psiopt$ is the value of $\Psi(k)$ evaluated at the optimal algorithm and is given by
\begin{multline}
\Psiopt\triangleq\sum_{i\in\sNR}\EEU{\Yik\lb \lambda_i q_i-{\Oneopt}\rb}+\\
X(k)\lb\sum_{i\in\sN}\frac{\EEU{\muRopt\PRopt}}{\Ts}-\Pavg\rb\\
+\sum_{i\in\sNNR}\Qik\lb \EEU{L\bRopt-\muRopt\Ropt}\rb,
\label{Psiopt}
\end{multline}
where $\PRopt$, $\muRopt$, $\Oneopt$ and $\Ropt$ satisfy \eqref{RT_QoS_Opt}, \eqref{P_avg_Opt} and \eqref{NRT_QoS_Opt}. Taking $\EE{\cdot}$ to \eqref{Psiopt}, summing over $k=0\cdots K-1$, dividing by $K$, taking the limit as $K\rightarrow\infty$ and using \eqref{RT_QoS_Opt}, \eqref{P_avg_Opt} and \eqref{NRT_QoS_Opt} gives
\begin{equation}
\limsup_{K\rightarrow\infty}\frac{1}{K}\sum_{k=0}^{K-1}\EE{\Psiopt}\leq 0
\label{Psi_Neg}
\end{equation}
Evaluating by Algorithm \ref{Scheduling_Alg} in the right-side of \eqref{Drift_Bound}, and taking $\EE{\cdot}$ with respect to $\bfU(k)$ to both sides gives
\begin{multline}
\frac{1}{2}\sum_{i\in\sNR}{\EE{Y_i^2(k)}}+\frac{1}{2}\sum_{i\in\sNNR}{\EE{Q_i^2(k)}}+\frac{1}{2}\EE{X^2(k)}\leq\\ C_1+\EE{\Psi^*(k)}.
\label{Drift_bound_Alg_1}
\end{multline}
Removing the two summations on the left-side of \eqref{Drift_bound_Alg_1}, summing over $k=0\cdots K-1$, dividing by $K$ then taking the limit as $K\rightarrow\infty$ yields
\begin{multline}
\limsup_{K\rightarrow\infty}\frac{\EE{X^2(K)}}{2K}\leq C_1+\lim_{K\rightarrow\infty}\frac{1}{2K}\sum_{k=0}^{K-1}\EE{\Psi^*(k)}\\
\overset{(a)}{\leq} C_1+\lim_{K\rightarrow\infty}\frac{1}{2K}\sum_{k=0}^{K-1}\EE{\Psiopt}\overset{(b)}{\leq}C_1.
\label{Drift_Constant}
\end{multline}
where inequalities (a) and (b) in \eqref{Drift_Constant} follow from the inequality $\Psi^*(k)\leq \Psiopt$ and \eqref{Psi_Neg}, respectively. Jensen's inequality says that $\EE{X(K)}\leq \EE{X^2(K)}$. Dividing by $K^2$, taking the square root, passing $K\rightarrow\infty$ and using \eqref{Drift_Constant} completes the mean rate stability proof. Similarly we can show the mean rate stability of $\Yik$.

\subsubsection{Objective Function Optimality} Evaluating the right-hand-side of \eqref{Drift_minus_Reward_Bound} at the optimal policy that has a genie-aided knowledge of the optimum reward $r_i(k)=\bRopt$ we get $\Delta(k)-\Bmax\sum_{i\in\sNNR}\EEU{Lr_i(k)}\leq C_1+\Psiopt-\Bmax\sum_{i\in\sNNR}\bRopt$ which is similar to equation (20) in \cite{li2011delay}. The optimality proof continues along the lines of Theorem 2 in \cite{li2011delay}.
\end{proof}

\section{Proof of Theorem \ref{Capacity_Thm}}
\label{App_Capacity_Thm_Proof}
\begin{proof}
We divide our proof into two parts. In the first part (Achievability), we show that if $\lambdaNR$ is strictly within the region $\sR$, then the queues can be stabilized. And the algorithm that stabilizes these queues is a modified version of Algorithm \ref{Scheduling_Alg_Cont}. We show this using the Lyapunov optimization technique \cite[pp.120]{srikant2013communication}. In the second part (Converse), we show that if $\lambdaNR\notin\sR$, then there exists no algorithm that guarantees the stability of the NRT queues.

\subsubsection{Achievability} We will show here that the following inequality holds under Algorithm \ref{Scheduling_Alg_Cont} which is the key to the proof.
\begin{multline}
\sum_{i\in\sNNR}\lambda_i\Qik + \sum_{i\in\sNR}\lambda_i q_i\Yik-\sum_{i\in\sN}X(k)\Pavg\leq \\
\E_{\bfU(k)}\left[\sum_{i\in\sNNR}\Qik\Dik+\sum_{i\in\sNR}\Yik\Dik\right]\\
-\E_{\bfU(k)}\left[\sum_{i\in\sN} \frac{X(k)\muRik\PRik}{\Ts}\right],
\label{Key_To_Proof}
\end{multline}
where $\Bik\triangleq\muRik\log\lb1+\PRik\gammaik\rb$. Once this inequality is proven, the rest of the achievability proof works similar to Theorem 5.3.2 in \cite[pp.120]{srikant2013communication}. Since $\lambdaNR (1+\epsilon)\in\sR$, to prove \eqref{Key_To_Proof} we multiply \eqref{NRT_Arr_Rate} by $\lambda_i$, \eqref{RT_Arr_Rate} by $\lambda_i$, and \eqref{P_avg_Capacity} by $(-\Pavg)$, then add the three inequalities after summing the first over $i\in\sNNR$ and the second over $i\in\sNR$ yielding
\begin{align}
\nonumber\sum_{i\in\sNNR}&\lambda_i\Qik + \sum_{i\in\sNR}\lambda_i q_i\Yik-\sum_{i\in\sN}X(k)\Pavg\leq \\
&\sum_{\bfm\in\sM^N}\lb\sum_{i\in\sNNR}\Qik\Dikm+\sum_{i\in\sNR}\Yik\Dikm\right.\\
&\left.-\sum_{i\in\sN} \frac{X(k)\muRikm\PRikm}{\Ts}\rb\pim
\label{Mult_and_Sum}\\
&\leq\sum_{\bfm\in\sM^N}\left[\sum_{i\in\sNNR}\PsiNRikst+\sum_{i\in\sNR}\PsiRikst\right]\pim,
\label{Mult_and_Sum_Bounded}
\end{align}
where $\Dikm\triangleq\muRik\log\lb1+\PRik\gammaik\rb$ while inequality \eqref{Mult_and_Sum} follows since the objective of problem \eqref{Max_Prob} is an upper bound on \eqref{Mult_and_Sum}. But since the right-hand-side of \eqref{Key_To_Proof} can be manipulated to give
\begin{align}
&\E_{\bfU(k)}\left[\sum_{i\in\sNNR}\Qik\Dikm+\sum_{i\in\sNR}\Yik\Dikm\right.\\
&\left.-\sum_{i\in\sN} \frac{X(k)\muRik\PRik}{\Ts}\right]\\
=&\E_{\bfU(k)}\left[\underset{i\in\sNNR}{\sum}\lb\Qik\Dikm-\frac{X(k)\muRik\PRik}{\Ts}\rb\right.\\
&\left.+\underset{i\in\sNR}{\sum}\lb\Yik\Dikm-\frac{X(k)\muRik\PRik}{\Ts}\rb\right]
\label{RHS_Manipulated}\\
=&\sum_{\bfm\in\sM^N}\left[\sum_{i\in\sNNR}\PsiNRikst+\sum_{i\in\sNR}\PsiRikst\right]\pim\geq\\
&\sum_{i\in\sNNR}\lambda_i\Qik +\sum_{i\in\sNR}\lambda_i q_i\Yik-\sum_{i\in\sN}X(k)\Pavg
\label{RHS_Compared}
\end{align}
where the left side of the inequality in \eqref{RHS_Compared} follows by evaluating \eqref{RHS_Manipulated} at Algorithm \ref{Scheduling_Alg_Cont} while its right side follows from \eqref{Mult_and_Sum_Bounded} which completes the proof of \eqref{Key_To_Proof}.

\subsubsection{Converse} The converse is done by showing that the upper bound of the sum of the number bits served from all NRT buffers under the best, possibly genie-aided, policy is less than the sum of bits arriving to the NRT buffers if the arrival rate does not satisfy \eqref{RT_Arr_Rate} through \eqref{PRikm_Positive}.\\
From the strict separation theorem \cite[pp.10]{srikant2013communication}, if $\lambda\notin\sR$ then there exists a vector $\beta\triangleq[\beta_1,\cdots\beta_{\NNR}]^{\rm T}\in\script{R}^{\NNR}$ and a constant $\delta>0$ such that for any vector $x\in\sR$ the following holds
\begin{equation}
\sum_{i\in\sNNR}\beta_i\lambda_i\geq\sum_{i\in\sNNR}\beta_ix_i+\delta
\label{Strict_Separation}
\end{equation}
Define $H(k+1)=H(k)+\sum_{i\in\sNNR}\beta_i\lb L\aik-\Bik\rb$ as the weighted sum of the queues where $\Bik\triangleq \muRik\Rik$ is the number of bits transmitted to user $i$ at slot $k$. Hence we have
\begin{equation}
H(K)=\sum_{k=0}^{K-1}\sum_{i\in\sNNR}\beta_i\lb L\aik-\Bik\rb.
\label{H_K}
\end{equation}
Define the set $\sK\triangleq\{k:m(k)=l, 0\leq k<K\}$ we can bound the second term in \eqref{H_K} as follows
\begin{align}
\nonumber \sum_{i\in\sNNR}&\beta_i\limsup_{K\rightarrow\infty}\sum_{k=0}^{K-1}\frac{\Bik}{K}\\
&\leq \sum_{i\in\sNNR}\beta_i\limsup_{K\rightarrow\infty}\sum_{l=1}^M\sum_{k\in\sK}\frac{\Bikst}{\vert\sK\vert}\frac{\vert\sK\vert}{K}
\label{Bounding_Bits}\\
&=\sum_{i\in\sNNR}\beta_i\sum_{l=1}^M\Bistl\pi_l=\sum_{i\in\sNNR}\beta_i\sum_{l=1}^M Lx_i\pi_l.
\label{increasing}
\end{align}
Adding $L\delta$ to both sides of \eqref{increasing} and using \eqref{Strict_Separation} yields
\begin{align}
\nonumber \sum_{i\in\sNNR}&\beta_i\limsup_{K\rightarrow\infty}\sum_{k=0}^{K-1}\frac{\Bik}{K}+L\delta \leq \\
\nonumber &L\lb\sum_{l=1}^M \pi_l\sum_{i\in\sNNR}\beta_ix_i+\delta\rb \leq  \sum_{i\in\sNNR}\beta_iL\lambda_i\\
&=\lim_{K\rightarrow\infty}\sum_{k=0}^{K-1}\frac{L\aik}{K}.
\label{Service_bounded_by_Arrival}
\end{align}
Combining \eqref{Service_bounded_by_Arrival} and \eqref{H_K} we conclude that $\limsup_{K\rightarrow\infty}H(K)=\infty$ which means that the weighted sum of the queues is unbounded, under the best possible policy, when $\lambdaNR\notin\sR$ which completes the proof.
\end{proof}



\bibliographystyle{IEEEbib}
\bibliography{MyLib}

\end{document}